\documentclass[
prl,
reprint,superscriptaddress
]{revtex4-1}
\usepackage{bbm}
\usepackage{xcolor}
\usepackage{enumitem}
\usepackage{amsmath}
\usepackage{amssymb}
\usepackage{amsthm}
\usepackage{graphicx}
\usepackage{hyperref}
\usepackage{cleveref}
\crefname{assumption}{assumption}{assumptions}

\newtheorem*{assumption}{Assumption}
\newtheorem{theorem}{Theorem}

\let\originalleft\left
  \let\originalright\right
\renewcommand{\left}{\mathopen{}\mathclose\bgroup\originalleft}
  \renewcommand{\right}{\aftergroup\egroup\originalright}

\newcommand{\id}{\mathbbm{1}}                
\newcommand{\expo}[1]{\operatorname{e}^{#1}} 
\newcommand{\tr}[1]{\operatorname{Tr}\left[ {#1} \right]} 
\newcommand{\trin}[1]{\operatorname{Tr}[ {#1} ]} 
\newcommand{\Tr}{\operatorname{Tr}} 
\newcommand{\braket}[1]{\vphantom{\left(#1\right)^A} \left \langle #1 \right \rangle } 
\newcommand{\ket}[1]{\left|#1 \right \rangle \vphantom{\left( #1 \right)^A}} 
\newcommand{\bra}[1]{\left\langle #1 \right | \vphantom{\left(#1\right)^A}} 

\newcommand{\dif}{\;\mathrm{d}}

\newcommand{\Real}{\mathbb{R}}

\newcommand{\Hil}{\mathcal{H}}

\newcommand{\supp}{\operatorname{supp}}
\newcommand{\rank}[1]{\relax\ifmmode\operatorname{rank}#1\else rank-$#1$\fi}

\IfFileExists{todonotes.sty}{%
  \usepackage[bordercolor=black,backgroundcolor=red!10!white]{todonotes}%
}{%

}

\allowdisplaybreaks[1]

\hypersetup{pdfborder={0 0 1}, bookmarks=true, linkcolor=blue}

\newcommand{\hi}{H^{\mathrm{int}}}
\newcommand{\thi}{\tilde{H}^{\mathrm{int}}}
\newcommand{\hs}{H_s}
\newcommand{\hu}{H_u}

\newcommand{\he}{H_e}
\newcommand{\hb}{H_b}

\newcommand{\hw}{H_w}

\newcommand{\hc}{H_c}

\newcommand{\timeorder}[1]{\mathcal{T} \left[ #1 \right]}
\newcommand{\oo}[1]{\mathcal{O}\left( {#1} \right)}
\newcommand{\subw}{{\raisebox{-5pt}{$\!\!{\scriptstyle w}\!$}}}
\newcommand{\subc}{{\raisebox{-4pt}{$\!\!{\scriptstyle c}$}}}
\renewcommand{\iint}{\int \!\!\!\!\! \int}

\begin{document}

\title{Clock-Driven Quantum Thermal Engines}

\author{Artur S.L. Malabarba}
\author{Anthony J. Short}
\affiliation{H.H. Wills Physics Laboratory, University of Bristol, Tyndall Avenue, Bristol, BS8 1TL, U.K.}
\author{Philipp Kammerlander}
\affiliation{Institute for Theoretical Physics, ETH Z\"{u}rich, Wolfgang-Pauli-Strasse 27, 8093 Z\"{u}rich, Switzerland}

\date{\today}

\begin{abstract}
    We consider an isolated autonomous quantum machine, where an explicit quantum clock is responsible for performing all transformations on an arbitrary quantum system (the engine), via a time-independent Hamiltonian.
    In a general context, we show that this model can exactly implement any energy-conserving unitary on the engine, without degrading the clock.
    Furthermore, we show that when the engine includes a quantum work storage device we can approximately perform  completely general unitaries on the remainder of the engine.
    This framework can be used in quantum thermodynamics to carry out arbitrary transformations of a system, with accuracy and extracted work as close to optimal as desired, whilst obeying the first and second laws of thermodynamics.
    We thus show that autonomous thermal machines suffer no intrinsic thermodynamic cost compared to externally controlled ones.
\end{abstract}
\maketitle

Recently there has been a great deal of interest in the application of thermodynamics to individual quantum systems, which may be composed of just a few atoms or qubits~\cite{HorodeckiOppenheim13,Skrzypczyk2014,Johan13,Skrzypczyk2014,Skrzypczyk13,AndersGiovannetti13,Alicki04,Aberg13,Frenzel14,Amikam12,Noah10,EitanRonnie96,Brandao13b,Brandao13,Hasegawa2010,Armen08}.
Given that thermodynamics was invented before quantum theory was even envisaged, and typically applies to macroscopic objects, it is perhaps surprising how close an analogy can be drawn between the quantum and classical case.
In~\cite{HorodeckiOppenheim13, Johan13,Skrzypczyk2014,Skrzypczyk13}, thermal engines are constructed out of quantum mechanical parts, incorporating an explicit system, thermal bath and work storage system.
In other approaches~\cite{AndersGiovannetti13,Alicki04}, the thermal engine is a system with externally-controlled Hamiltonian and access to a thermal bath.

So far, these frameworks all involve the external application of discrete transformations to the thermal engine.
An interesting open question, raised by several authors~\cite{Skrzypczyk2014,HorodeckiOppenheim13,Aberg13,Frenzel14}, is whether this external control should carry a thermodynamic cost, and how to include this control explicitly in the framework.
In this paper, we address this issue by describing how an explicit quantum clock can control the evolution of a completely arbitrary quantum engine, thus allowing any unitary protocol to be carried out via a time-independent global Hamiltonian.

We first show that any  energy-conserving unitary operation can be exactly implemented on a quantum system (the engine) by attaching a quantum clock to it via the correct time-independent interaction Hamiltonian.
Furthermore, this process is essentially independent of the initial state of the clock, requiring only that it lies within a known finite region.
In particular, it is not necessary for the clock to precisely specify the `time'.
After the unitary has been fully implemented, the clock is not correlated with the system and could be used to perform further operations.

Next, we show that we can also approximately implement any unitary on a quantum system, including those which change the energy, by including an explicit work storage system in the engine (essentially a `weight on a string').
We can achieve arbitrarily good accuracy by using a weight with a sufficiently narrow momentum distribution.

Finally, we consider this framework in the context of quantum thermodynamics.
We show that our clock-driven engine obeys the first and second laws of thermodynamics, and that any transformation of the system can be implemented to any desired accuracy whilst extracting work as close as desired to the reduction in free-energy of the system.
Furthermore, after an optimal protocol neither the clock nor the weight are degraded relative to their initial states in their power to carry out subsequent transformations.
We thus show that clock-driven thermal engines suffer no intrinsic thermodynamic cost compared to externally controlled ones.

We end with a discussion on the viability of alternative clocks, average energy conservation, and the use of the clock also to measure time in the system.

\section{The Clock}
\label{sec:clock-framework}

We first consider a Hilbert space divided into $2$ parts, the clock and the engine, $\Hil = e \otimes c$.
The engine $e$ is an arbitrary quantum system that can start in any initial state $\rho_e$ and be subject to any Hamiltonian $\he$.
As an example, it could be a simple qubit system or it could be a sophisticated thermal machine composed of several high-dimensional subsystems.

The clock, $c$, is a way of controlling the evolution of the engine, without having to provide external input.
In our idealized framework, it has the continuous-spectrum Hamiltonian $\hc = v P_c$ where $P_c$ is the momentum operator in $c$ and, for convenience, we take $v=1 m s^{-1}$.
Under its free evolution, the clock state, $\rho_c$, thus moves to the right with constant unit velocity, and one may interpret its position $X_c$ as reflecting time.

The engine and the clock interact through a static Hamiltonian $\hi$ by means of which the control is implemented.
The total Hamiltonian thus takes the form
\begin{equation}
    \label{eq:4}
    H = \he \otimes \id_c + \id_e \otimes \hc + \hi.
\end{equation}
Note that all of these Hamiltonians are time-independent.
For conciseness, we omit the identities here on out.

In order to accurately and repeatably implement unitary operations on the engine, we only need a simple assumption on the initial state.
\begin{assumption}
    \label{as:3}
    The initial state is a product state
    \begin{equation}
        \label{eq:37}
        \rho(0) = \rho_e(0) \otimes \rho_c(0),
    \end{equation}
    and the support of $\rho_c(0)$ in position is contained inside a known finite interval.
    We define $K$ as the size of this interval.
\end{assumption}
Knowing that the clock is located inside some region is \emph{all} we require of it, which is in stark contrast to the stronger assumption of the clock being initially in a very narrow position state (corresponding to a well-defined `time') which one could have expected here.

This is a mathematically convenient assumption, which simplifies the calculations at very little cost.
Any normalizable state is always close in trace distance to a state with finite support.
Since the steps involved in the proofs below never increase the trace distance, all results will hold up to arbitrary precision for any clock state, even if it has infinite tails---say, a Gaussian distribution.

\subsection{Energy-conserving Unitaries}
\label{sec:role-clock}

We now constructively show that it is possible within our framework to exactly implement any energy-conserving unitary $U_e$ on the engine via interactions with the clock.
This succeeds \emph{for any} initial state satisfying \cref{eq:37}, and the clock and engine are always unentangled at the end.

We start by choosing an interaction Hamiltonian of the form
\begin{equation}
    \label{eq:3}
    \hi = \int_\Real \hi_{e}(x) \otimes \ket{x}_\subc \!\bra{x} \dif x,
\end{equation}
where $\ket{x}_\subc$ are the clock's position eigenstates.
Since the clock's position increases linearly with time, this means the clock is driving the engine by applying on it an effective time-dependent Hamiltonian.

We then choose $\hi_{e}(x)$ such that
\begin{equation}
    \label{eq:1}
    [\hi,\he ] = 0,
\end{equation}
thus guaranteeing that the interaction will never transfer energy between the clock and the engine.
This also prevents the clock and engine from becoming entangled.

Given this commutation relation, $\hi$ described in the interaction picture takes the form
\begin{align}
  \label{eq:20}
  \thi(t) &= \expo{\frac{i}{\hbar}(\he + \hc) t} \hi \expo{-\frac{i}{\hbar} (\he + \hc) t} \\
          &= \int_\Real \hi_{e}(x) \otimes \ket{x-t}_\subc\!\bra{x-t} \dif x. \notag
\end{align}
Thus, shifting $x\to x+t$, one can see that the total evolution operator between times $0$ and $t$ is given by
\begin{align}
    \label{free_engine}
  \mathcal{U}(t) &= \expo{-\frac{i}{\hbar}  (\he + \hc) t}
                   \timeorder{ \expo{- \frac{i}{\hbar} \int_{0}^{t} \thi(t') \dif t'}} \\
                 &= \expo{-\frac{i}{\hbar} (\he + \hc) t}
                   \!\int_\Real\! \timeorder{ \expo{ -\frac{i}{\hbar} \int_{0}^{t} \hi_{e}(x+t') \dif t'}}
                   \otimes \ket{x}_\subc\!\bra{x} \dif x \notag \\
                 &= \expo{-\frac{i}{\hbar} \he t} \int_\Real
                   \timeorder{ \expo{ -\frac{i}{\hbar} \int_{0}^{t} \hi_{e}(x+t') \dif t'}}
                   \otimes \ket{x+t}_\subc\!\bra{x} \dif x, \notag
\end{align}
where $\timeorder{\cdot}$ is the time ordering operation.

\begin{figure}
    \centering
    \includegraphics[width=\linewidth]{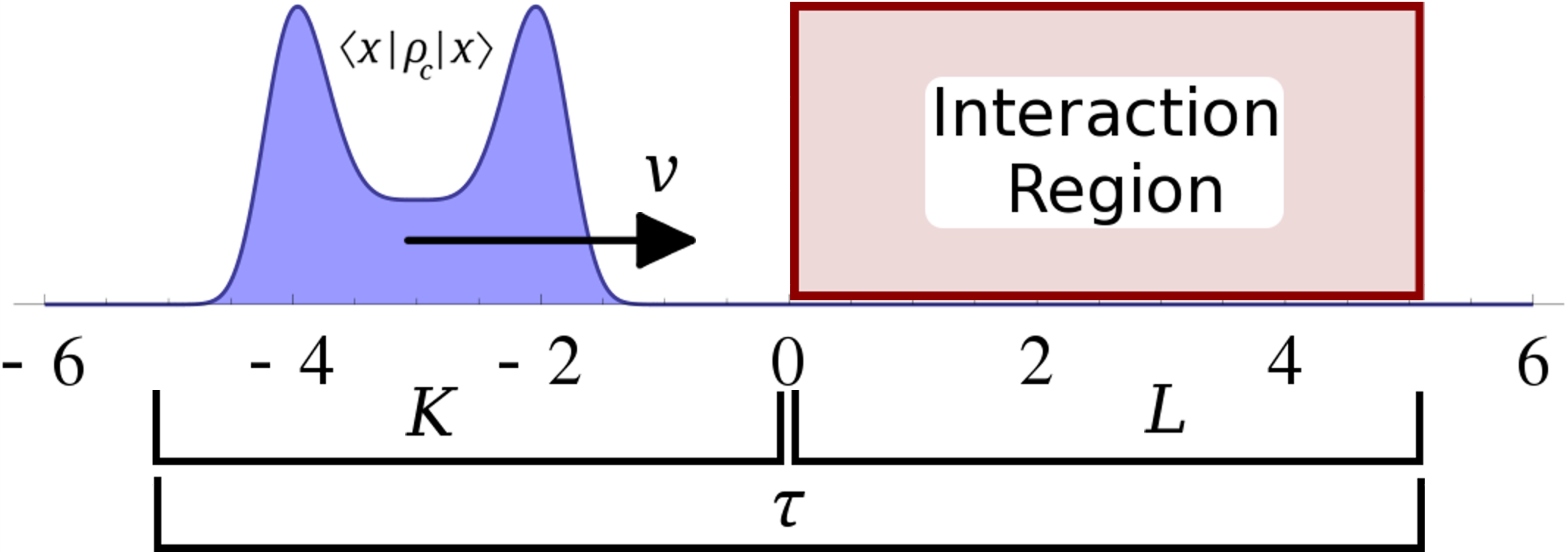}
    \caption{
      Graphical diagram of a possible initial clock state (left) and the interaction region (the support in position of $\hi_e(x)$, right).
      With $v=1$, for $t>\tau$ the clock state will be entirely to the right of the interaction region.
    }
    \label{fig:1}
\end{figure}
Note that the clock moves forward at constant speed, despite the interaction.
We now choose $\hi_e$ to have support inside an interval of size $L$ immediately to the right of the clock state's support (the interaction region), define $\tau = K + L$ and choose $t > \tau$ so the clock has time to completely `cross-over' the support of $\hi_e$ (see \cref{fig:1}).
This causes the time integral in \cref{free_engine} to go over the entire support of $\hi_e$, becoming independent of $x$.

In other words, denoting as $\supp[f(x)]$ the support of $f$ on $x$, we have $\supp[\hi_e(x)] \subset [x',x'+\tau]$ for all $x' \in \supp[\braket{x|\rho_c(0)|x}]$.
And so, for all $t>\tau$, the time evolution acts separately on the clock and on the engine
\begin{align}
  \label{eq:23}
  \rho(t) &= \mathcal{U}_e(t)\, \rho_e(0)\ \mathcal{U}_e^\dagger(t) \otimes
            \expo{-\frac{i}{\hbar} \hc t}\rho_c(0)
            \expo{\frac{i}{\hbar} \hc t} \notag\\
  \mathcal{U}_e(t) &= \expo{-\frac{i}{\hbar} \he t} U_e \\
  U_e &= \timeorder{ \expo{ -\frac{i}{\hbar}
        \int_{\supp(\hi_e)} \hi_{e}(t') \dif t'}}, \notag
\end{align}
which proves that the engine is acted on by $U_e$ and then undergoes free evolution indefinitely.

Finally, for any unitary $V$, there is always a Hermitian operator acting on the space, $G_{V}$, such that $V = \expo{-iG_V}$.
Thus, one needs simply to choose $\hi_{e}$ such that $\int_{\supp(\hi_e)}\hi_{e}(t') \dif t' = G_{U_{e}}$ in order to implement any desired energy-conserving $U_{e}$.
One possibility is a fixed Hamiltonian which switches on and off, i.e. $\hi_e(t) = i f(t)\!\ln U_e$, where $f$ is a normalized function with support inside the interaction region.
However, note that our approach applies to any form of time dependence, which incorporates a broader range of experimental procedures.

Of course, the clock and engine can become entangled during the procedure, but they will always be in a product state at the end ($t>\tau$).
In fact, the state of the clock doesn't change other than being translated.
This means the operation never degrades the clock.

\section{The Weight}
\label{sec:weight-framework}

Here, we show that one can extend the above considerations to implement general, not necessarily energy-conserving, unitaries on a subsystem of the engine, by compensating on the rest of it.
For that, consider the engine to be composed of two parts, $e = s\otimes w$.

The system, $s$, is the part which we wish to transform.
It is finite-dimensional and has arbitrary Hamiltonian $\hs$ and initial state $\rho_s(0)$.
Our objective is that, after applying $U_e$ on $\rho_s(0)\otimes\rho_w(0)$, the state of the system should be close to $V_s \rho_s(0) V_s^\dagger$, where $V_s$ is a general unitary on $s$.

The weight, $w$, acts as an energy-storage device, which can be raised or lowered to extract or supply energy~\cite{Skrzypczyk2014}.
It has the continuous-spectrum Hamiltonian $\hw = M g X_w$~\footnote{The continuous-spectrum Hamiltonian was chosen for convenience. It has been shown in~\cite{Johan13,Skrzypczyk13} that discrete weights also work arbitrarily well.}, where we choose $M = 1 J m^{-1}/g$.
The primary purpose of the weight is to compensate the system's energy change, since the final state of $s$ can have a different energy than the initial state.
However, as we show below, the initial state $\rho_w(0)$ of the weight may limit how optimally we can transform a state.
In particular, to transform non-diagonal states optimally we need the initial state of the weight to be narrow in momentum and be centered around a known $p_0$ so it can also act as a resource of coherence.
{\AA}berg~\cite{Johan13} pointed out that this resource can be used catalytically, and we show that this holds in our framework.

\subsection{Arbitrary Transformations}
\label{sec:arbitrary-unitaries}

Here, we constructively prove that for any unitary $V_s$ on $s$ there is a choice of $\hi$ which approximately implements it.
First, we choose $U_e$ to have the form
\begin{equation}
    \label{eq:22}
    U_e = \sum_{n,j}
    \braket{E^s_n | V_s | E^s_j} \ket{E^s_n}\!\langle{E^s_j}|\otimes
    \expo{-\frac{i}{\hbar} (P_w - p_0)  (E^s_j - E^s_n) },
\end{equation}
where $E^s_j$ are the energy levels of $\hs$, and $|{E^s_j}\rangle$ are their respective eigenstates.
Note that this unitary satisfies energy-conservation.
Furthermore, it is translation-invariant on $w$, i.e.,
\begin{equation}
  \label{eq:7}
  [U_e, P_w] = 0,
\end{equation}
where $P_w$ is the weight's momentum operator.
This serves two purposes: (\emph{i}) it guarantees the protocol works regardless of the initial height of the weight (or how much energy is stored in it), (\emph{ii}) it enables the catalytic use of the coherences in the weight to transform the system~\cite{Johan13}.
By the proof above, in order to implement this $U_e$, we may choose an $\hi$ which also satisfies $[\hi, P_w] = 0$.
In particular, one possible example is the aforementioned $\hi_e(t) = if(t)\!\ln U_e$.

Alternatively, \cref{eq:22} can be written as
\begin{align}
  \label{eq:47}
  U_e &= \int_\Real U_s(p) \otimes\ket{p}_\subw\!\bra{p} \dif p \\
  U_s(p) &= \sum_{n,j} \expo{-\frac{i}{\hbar} (p - p_0)(E^s_j - E^s_n)}
           \braket{E^s_n | V_s | E^s_j} \ket{E^s_n}\!\bra{E^s_j}, \notag
\end{align}
where $\ket{p}_w$ are the momentum eigenstates of the weight.
In this form, one sees that, upon applying $U_e$ on the engine state, the reduced state of the system becomes
\begin{align}
    \label{eq:48}
    \!\!\Tr_{w} \left[ U_e \rho_e(0) U_e^{\dagger} \right] &= \int_\Real \mu_w(p)\, U_s(p)\rho_s(0) U_s^\dagger(p) \dif p,
\end{align}
where $\mu_w(p) = \bra{p} \rho_w(0) \ket{p}$ is the initial momentum distribution of the weight.

As shown in the Appendix, this can be made arbitrarily close to $V_s \rho_s(0) V_s^{\dagger}$ in trace distance, by making $\rho_w(0)$ narrow enough in momentum space.
Intuitively speaking, the closer $\mu_w(p)$ is to a delta function, the closer this operation is to the exact $V_s$ (which is $U_s(p_0)$).
Therefore, for any $\epsilon>0$, there is always a good enough $\rho_w(0)$ such that
\begin{equation}
    \label{eq:32}
    \left\| \Tr_{w} \left[ U_e \rho_e(0) U_e^{\dagger} \right] - V_s \rho_s(0) V_s^{\dagger} \right\|_1 \leq \epsilon.
\end{equation}
Moreover, since $[U_e,P_w]=0$, $\mu_w(p)$ is conserved by the operations.
Also, note that the full evolution includes both the operations and the free evolution, but the weight's free evolution only shifts $p_0$ without affecting the shape of $\mu_w(p)$.

As the error in implementing $V_s$ depends only on how narrow $\mu_w(p)$ is, this implies that the usefulness of the weight is not degraded.

An interesting special case arises if $V_s$ only permutes energy levels, $V_{s} |{E_j^{s}}\rangle = | E_{\pi(j)}^{s}\rangle$ for some permutation $\pi$, and if the initial state is diagonal in energy, $\rho_{s}(0) = \sum_{n} p_n^{s}| {E_n^{s}}\rangle\!\langle{E_n^{s}}|$.
Joining these two, \cref{eq:48} simplifies to
\begin{align}
  \label{eq:9}
  \Tr_{w}& \left[ U_e \rho_e(0) U_e^{\dagger} \right] \notag\\
         &= \sum_{n,m} |{E^{s}_{\pi(m)}}\rangle\!\bra{E^{s}_m} \rho_{s}(0) \ket{E^{s}_n}\!\langle{E^{s}_{\pi(n)}}| \notag\\
         &\quad\quad\int_\Real \expo{-\frac{i}{\hbar} (p-p_0)(E^{s}_n - E^{s}_{\pi(n)})}\mu_w(p)
           \expo{\frac{i}{\hbar} (p-p_0)(E^{s}_m - E^{s}_{\pi(m)})} \dif p. \notag\\
         &= \sum_{n}p_n^{s} |{E^{s}_{\pi(n)}}\rangle\!\langle{E^{s}_{\pi(n)}}| \notag\\
         &= V_{s} \rho_{s}(0) V_{s}^\dagger,
\end{align}
which means the transformation can be performed exactly regardless of the state of the weight.

\section{Work Cost of Transformations}
\label{sec:framework}

The results above are of very general application in the field of quantum thermodynamics, more specifically in quantum resource theories.
We exemplify this by answering a question posed in~\cite{Skrzypczyk2014,HorodeckiOppenheim13}: ``\emph{Should unitary operations pose a thermodynamic cost during work extraction?}''

\begin{figure}
    \centering
    \includegraphics[width=\linewidth]{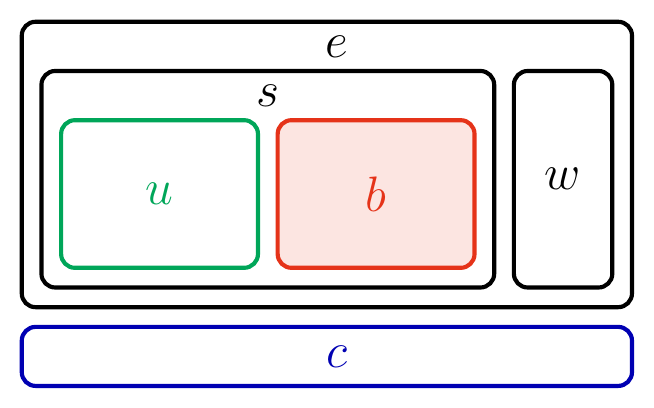}
    \caption{Representation of how the Hilbert space is divided into clock $c$, engine $e$, weight $w$, system $s$, subsystem $u$, and bath $b$.}
    \label{fig:2}
\end{figure}
As such, we further divide the system into two parts, $u \otimes b$, so that the total Hilbert space is $\Hil = u \otimes b \otimes w \otimes c$ (see \cref{fig:2}).
The working subsystem $u$ has arbitrary Hamiltonian $\hu$ and finite dimension $d_u$.

The thermal bath, $b$, is composed of an arbitrary number of finite dimensional systems with arbitrary Hamiltonians in a thermal state at temperature $T$.
Any protocol must specify which bath systems it will use, with their combined Hamiltonian given by $\hb$, and the initial state being $\rho_b(0) = \expo{-\frac{\hb}{k T}}/\trin{\expo{-\frac{\hb}{k T}}}$.

The weight, $w$, follows the same rules as above, but gains a new importance.
Its change in average energy now also represents the thermodynamic work cost or gain.

Using this structure, we can explicitly define the thermodynamic quantities such as internal energy change, extracted work and emitted heat up to time $t$, where $t=0$ is the time when the protocol starts:
\begin{align}
    \label{eq:24}
    \Delta U &= \tr{\hu(\rho(t) - \rho(0))} \notag\\
    W &= \tr{\hw(\rho(t) - \rho(0))} \notag\\
    Q &= \tr{\hb(\rho(t) - \rho(0))}.
\end{align}
Following~\cite{Skrzypczyk2014}, we also define the free energy of a state by $F(\sigma)= \tr{H \sigma} - T S(\sigma)$ where $S(\sigma)$ is the von Neumann entropy of $\sigma$.
With this, we now show that the definitions above obey the first and second laws of thermodynamics, and that optimal work extraction is possible.

\subsection{First and Second Laws}
\label{sec:first-second-laws}

Consider a quantum machine whose Hilbert space and Hamiltonian are as described above, and take as assumptions the following identities:
\begin{subequations}
    \label{eq:39}
    \begin{equation}
        [\hi,\he ] = 0, \label{eq:26}
    \end{equation}
    \vspace{-0.9cm}
    \begin{equation}
        \label{eq:21}
        [\hi,X_c ] = [\hi,P_w ] = 0,
    \end{equation}
\end{subequations}
and $\rho(0) = \rho_u(0)\otimes \rho_b(0)\otimes \rho_w(0)\otimes \rho_c(0)$.
Note that these are all satisfied by the constructions thus far.

The first law of thermodynamics is trivially satisfied by \cref{eq:26}, which implies that $\braket{\he}$ is a conserved quantity\footnote{In fact, the expected value of any function of $\he$ is conserved.}, and so
\begin{align}
  \label{eq:25}
  \Delta\braket{\hu} + \Delta\braket{\hw}
  + \Delta\braket{\hb} &= \Delta\braket{\he} = 0 \notag\\
  \Rightarrow\quad\quad\Delta U + W + Q &= 0.
\end{align}

To prove the second law, we construct a Kelvin-Planck statement~\cite{Thomson51, Planck97}, showing that it is impossible to extract positive work in a cyclic process.
From \cref{eq:21}, we know $\hi$ must satisfy
\begin{equation}
    \label{eq:38}
    \hi = \iint_{\Real^2} \hi_{ub}(x,p) \otimes
    \ket{p}_\subw\!\bra{p}\otimes \ket{x}_\subc\!\bra{x}\dif p \dif x.
\end{equation}
As detailed in the Appendix, this means the reduced state on $u\otimes b$ is a mixture of unitaries $V(x,p)$ applied on the initial state,
\begin{equation}
  \label{eq:27}
  \!\!\!\!\!\Tr_{wc}[\rho(t)]
  \!=\!\!\! \iint_{\Real^2}\!\!\!\!\! \mu(x,p) V(x,p,t) \rho_{ub}(0) V^\dagger(x,p,t) \!\!\dif p \!\dif x, \!\!\!
\end{equation}
where $\mu(x,p) = \braket{x,p|\rho_{cw}(0)|x,p}$.
This means the entropy of $u\otimes b$ never decreases,
\begin{equation}
    \label{eq:28}
    0 \leq \Delta S(\rho_{ub}) 
    \quad \mathrm{or} \quad  -\Delta S(\rho_{b}) \leq \Delta S(\rho_{u}) 
\end{equation}
where $\Delta$ denotes a difference between times $0$ and $t$, and we have used the subadditivity of the entropy and the assumption that the initial state is a product state.

Since the initial bath state has minimum free energy for a given temperature, one finds
\begin{align}
  \label{eq:30}
  &&0 \leq \Delta F(\rho_b) &= \Delta E(\rho_{b}) - T \Delta S(\rho_{b})  \\
  \Longrightarrow\!\!\!\!&& T \Delta S(\rho_{b})&\leq \Delta E(\rho_{b}) = - \Delta E(\rho_{w}) - \Delta E(\rho_{u}) \notag\\
  \Longrightarrow\!\!\!\!&& W = \Delta E(\rho_{w}) &\leq -T \Delta S(\rho_{b}) - \Delta E(\rho_{u})  \leq - \Delta F(\rho_{u}) \notag
\end{align}
This means one cannot extract more energy than the reduction in free energy of the subsystem.
Clearly, if the thermodynamic properties of the subsystem are the same in its initial and final state, $\Delta F(\rho_u) = 0$ and hence positive work cannot be extracted from the bath.
That is, \emph{``One cannot turn heat purely into work''}.

Note that this result is a direct consequence of the assumptions in \cref{eq:39}, and not of any implementation details.
So the first and second laws hold given an arbitrary protocol which follows these assumptions, even one which is far from optimal, or which acts on a different initial state from the one it was designed for.

\subsection{Optimal Transformations}
\label{sec:work-extraction}

We now show that this framework allows thermodynamically optimal state transformation protocols.
In particular, a protocol exists such that the work extracted is arbitrarily close to the reduction in free energy of the subsystem, and the final state of the subsystem is arbitrarily close to the desired final state.
Here, a protocol represents a unitary operation on $u\otimes b\otimes w$ or, equivalently, an $\hi$ on $u\otimes b\otimes w\otimes c$.

This section is based on the protocols in~\cite{Skrzypczyk2014,Skrzypczyk13}, but our framework is slightly different and involves a simpler proof strategy.
Given the above results, we need only find a unitary $V_s$ on $s = u\otimes b$ with the desired effect.
Then, the existence of an interaction Hamiltonian which implements this unitary relies only on the weight state being good enough.

Let us write the initial state of the subsystem as
\begin{align}
  \label{eq:31}
  \rho_u(0) &= \sum_{n} p_n \ket{\psi_n}\!\bra{\psi_n},
\end{align}
and the desired target state of the subsystem by
\begin{align}
    \sigma_u = \sum_m q_m \ket{\phi_m}\!\bra{\phi_m},
\end{align}
where $\{\ket{\psi_n}\}$ and $\{\ket{\phi_n}\}$ are respective eigenbases, labeled so that $p_n>p_{n+1}$ and $q_n>q_{n+1}$.
Thus, we desire that $\rho_u(t) = \expo{-\frac{i}{\hbar} \hu t} \sigma_u \expo{\frac{i}{\hbar} \hu t}$, $\forall t > \tau$, so the interacting initial state becomes the freely-evolving target state.

For simplicity we assume that $\sigma_u$ is full rank.
If the final state has lower rank, we can instead transform the subsystem to a state $\sigma_s'$ with full rank and close to $\sigma_u$ in trace distance.

The desired subsystem-bath unitary, $V_{ub} = V_s$, is composed of three stages.
The first stage acts only on $u$, rotating $\rho_u$ to be diagonal in its energy basis.
That is,
\begin{align}
    \label{eq:45}
  V_{ub}^{(1)} &= \sum_{n} \ket{E^u_n} \! \bra{\psi_n} \otimes \id_b,
\end{align}
where $\ket{E^u_n}$ are the eigenstates of $\hu$.
After this step the state of the subsystem will be
\begin{align}
    \sum_{n} p_n \ket{E^u_n}\!\bra{E^u_n}.
\end{align}
Note that entropy is conserved in this stage, so the total energy change is equal to the change in free energy of $\rho_u$.

The second stage, $V_{ub}^{(2)}$, takes this diagonal state into another diagonal state, and has been studied before~\cite{Skrzypczyk13,Skrzypczyk2014,Johan13,HorodeckiOppenheim13}.
We present here a proof of its energy efficiency which is simpler than previous ones.
The unitary is divided into many small steps, acting on $u\otimes b$ and gradually changing the probabilities of the subsystem's energy levels from $p_n$ to $q_n$, by performing the swap operation on the subsystem state and a similar state from the bath.
The first step uses a thermal state in the bath of dimension $d_u$, whose probabilities $p_n'$ lie between $p_n$ and $q_n$ and are very close to the former, i.e., $\left| p_n - p_n' \right|<\delta p$ for all $n$.

Since this thermal state has minimum free energy, $\Delta F_b$ must be of order $\delta p^2$, so $\Delta E_b = T \Delta S_b + \oo{\delta p^2}$.
Furthermore, due to the swap operation, $\Delta S_b = - \Delta S_u$, and so the energy variation of this step is
\begin{align}
    \label{eq:44}
  \Delta E^1_{ub} &= \Delta E^1_b + \Delta E^1_u = \Delta F^1_u + \oo{\delta p^2}.
\end{align}

Since $|p_n - q_n| < 1$ for all $n$, we can choose a unitary $V_{ub}^{(2)}$ such that only $\delta p^{-1}$ steps are necessary, and so the total energy variation must be
\begin{align}
    \label{eq:44}
  \Delta E_{ub} &= \sum_{j = 1}^{\delta p^{-1}} \Delta F_u^j + \oo{\delta p^2} = \Delta F_u + \oo{\delta p}.
\end{align}
After this stage the state of the subsystem will be exactly
\begin{align}
    \sum_{n} q_n \ket{E^u_n}\!\bra{E^u_n}.
\end{align}
Note that this step only consists of permutations on the $\hu+\hb$ energy basis, which means that if both $\rho_u(0)$ and $\sigma_u$ are diagonal then the transformation can be performed perfectly regardless of the state of the weight.

In the final stage, we again act only on the subsystem, rotating it into the eigenstates of $\sigma_u$ via the unitary
\begin{align}
    \label{eq:45}
    V_{ub}^{(3)} &= \sum_{n} \ket{\phi_n} \! \bra{E^u_n} \otimes \id_b,
\end{align}
which takes the final state of the subsystem into $\sigma_u$,
\begin{align}
    \label{eq:17}
    \Tr_b \left[ V_{ub} \rho_{ub}(0) V_{ub}^\dagger \right] = \sigma_u,
\end{align}
with $V_{ub} =  V^{(3)}_{ub} V^{(2)}_{ub} V^{(1)}_{ub}$.
Note that the protocol also involves free evolution, see \cref{eq:23}, so the subsystem's state becomes the freely-evolving $\sigma_u$.
That is, $\rho_u(t) = \expo{-\frac{i}{\hbar} \hu t} \sigma_u \expo{\frac{i}{\hbar} \hu t}$ for all $t > \tau$.
If one desires to have exactly $\sigma_u$ at a specific time $t'$, one simply needs to add a fourth step to the unitary, $V_{ub} =  \expo{i \hu (t-t')}V^{(3)}_{ub} V^{(2)}_{ub} V^{(1)}_{ub}$.

Looking at the energy balance, we find that the energy change of the subsystem and bath is as close as desired to the free energy change of the subsystem, making it a thermodynamically optimal transformation.
That is,
\begin{align}
    \label{eq:46}
    \Delta E_{ub} = \Delta F_{u} + \epsilon,
\end{align}
where $\epsilon > 0$ can be made as small as desired.

By the results above, if the state of the weight is narrow enough in momentum, there exists an $\hi$ which implements $V_{ub}$ after a time $\tau$, so that $\left\| \Tr_{bwc}\rho(t) - \expo{-\frac{i}{\hbar} \hu t} \sigma_u \expo{\frac{i}{\hbar} \hu t} \right\|_1$ is small for all $t >\tau$.
This means \cref{eq:46} still holds up to a small error and, by energy conservation, this energy difference must have been transferred to the weight.

Thus for any upper bound $\epsilon'>0$ we wish to impose on the error, there is always a good enough weight state (small trace distance) and a gradual enough protocol (small $\delta p$) such that, by our definition of work,
\begin{align}
  \label{eq:33}
  W &= \tr{H_w(\rho(\tau) - \rho(0))} \notag\\
    &\geq - \Delta F_u - \epsilon'.
\end{align}

\section{Discussion}
\label{sec:discussion}

Here, we have shown it is possible to exactly perform any energy-conserving unitary on a closed quantum system by attaching it to a quantum clock via a static interaction Hamiltonian.
We have extended this so that any unitary can be approximated by attaching this system to a weight.
Furthermore, this framework was proven to always satisfy the laws of thermodynamics, and to be a viable implementation of optimal quantum thermal machines.
It was also shown that neither the clock nor the weight are degraded by the procedure, so they can be repeatedly used to transform a succession of systems.

This addresses the question of whether quantum thermal machines, composed of nothing but time-independent Hamiltonian evolution, can be as efficient as externally controlled ones.
Remarkably, even if the clock has a broad initial state, this poses no thermodynamic cost in principle.
There is also no intrinsic cost in using the weight, albeit there is a stronger restriction on the initial state.
Namely, one needs it to be narrow in momentum space in order to achieve arbitrary precision.

For simplicity, we have considered the domain in the clock's position space to be $\Real$, but the same results can be achieved with a periodic clock as long as the timescales involved are smaller than the period.
It is an open question whether one can derive similar results with a more physical Hamiltonian, such as one whose energies are bounded from below or one which is finite dimensional.
One way of doing so could be to look for physical Hamiltonians $H_c$ that are similar to $v P_c$ whenever the state of the clock lies in a certain region of the state space.
In this way, it seems like it should be possible to approximate our results sufficiently well, by choosing the interaction Hamiltonian and initial state of the clock such that the system stays inside this region with high probability.

So far, we used the clock as a quantum mechanical way of controlling the engine, not as a quantum time-measuring device.
If one desires to use it as such, the unitary is still implemented exactly, but there are two relevant scenarios to consider with regards to knowing the state of the subsystem.
For instance, consider the initial state of the clock to have support inside $[-K, 0]$.
(1) If the desired final state of the subsystem $\sigma_u$ is diagonal in its energy basis, then we are guaranteed to have this state if measuring the clock's position yields a value greater than $\tau$.
(2) If the final state is not diagonal, then measuring a value of $x > \tau$ for the clock's position indicates we have a state between $\expo{- \frac{i}{\hbar} \hu (x+K)} \sigma_u \expo{ \frac{i}{\hbar} \hu (x+K)}$ and $\expo{-\frac{i}{\hbar} \hu {x}} \sigma_u \expo{\frac{i}{\hbar} \hu {x}}$, so our precision is dependent on how narrow the clock state is.
In addition, having to perform a measurement could impose additional thermodynamic costs.

Throughout this work, we have considered unitaries and interaction Hamiltonians which commute with the free Hamiltonian of the engine. 
Previous work~\cite{Skrzypczyk2014} also allowed unitaries which preserve the average energy of the specified initial state, without commuting with the Hamiltonian.
This allowed for optimal protocols which were independent of the state of the weight.
However, in order to implement such unitaries with the clock it seems that one would need to place strong conditions on its initial state. 
This is one of the reasons why we chose interactions which commute with the engine Hamiltonian.

\section{Acknowledgements}
\label{sec:aknowledgements}

The authors gratefully acknowledge fruitful discussions with Lea Kr\"{a}mer Gabriel, Ralph Silva, Renato Renner, Sandra Rankovic, and Sandu Popescu.
AJS acknowledges support from the Royal Society.
ASLM acknowledges support from the Conselho Nacional de Desenvolvimento Cient\'{i}fico e Tecnol\'{o}gico.
PK acknowledges support from the Swiss National Science Foundation (through the National Centre of Competence in Research `Quantum Science and Technology') and the European Research Council (grant 258932).
This work was partially supported by the COST Action MP1209.

\bibliography{references}
\appendix

\section{\Large Appendix}
\subsection{Close in Trace Distance}
\label{sec:close-trace-distance}

To say that a statement is true for $\mu_w(p)$ narrow enough and centered around $p_0$ is equivalent to saying it is true for $\frac{1}{\delta} \mu_w(\frac{p-p_0}{\delta} + p_0)$ with a small enough $\delta > 0$.
\begin{theorem}
    \label{thr:1}
    For any finite dimensional Hilbert space $\Hil$, let $\mathcal{D}(\Hil)$ be its set of density matrices.
    Given any $\epsilon >0$, any probability distribution $\mu_w:\Real\to [0,\infty)$ with a well defined first moment, and any continuous function $\nu:\Real\to\mathcal{D}(\Hil)$, there is always a $\delta>0$ such that
    \begin{equation}
        \label{eq:49}
        \left\| \frac{1}{\delta } \int_\Real \mu_w \left(\frac{p-p_0}{\delta} + p_0 \right)
          \nu(p) \dif p - \nu(p_0) \right\|_1
        \leq \epsilon,
    \end{equation}
    where $p_0 = \int p\, \mu_w(p) \dif p$.
\end{theorem}
\begin{proof}
    By virtue of the continuity of $\nu$, there is always a $\delta'$ such that
    \begin{equation}
        \label{eq:34}
        \max_{ p \in \Delta} \left\|\nu(p) - \nu(p_0) \right\|_1
        \leq \frac{\epsilon}{2},
    \end{equation}
    where $\Delta = (p_0-{\delta'},p_0+{{\delta'} })$.
    Furthermore, since $\mu_w$ is a normalized probability distribution on $\Real$, for any $\delta'>0$ there is always a $\delta$ such that
    \begin{equation}
        \label{eq:35}
        \int_{p_0- \frac{\delta'}{\delta}}^{p_0+\frac{\delta'}{\delta}} \mu_w \left(p \right) \dif p
        > 1- \frac{\epsilon}{4}.
    \end{equation}
    Given that
    \begin{align}
      \label{eq:2}
      \int_{p_0- \frac{\delta'}{\delta}}^{p_0+\frac{\delta'}{\delta}} \mu_w \left(p \right) \dif p &=\int_{- \delta'}^{\delta'} \mu_w \left( \frac{p}{\delta} + p_0 \right) \frac{\dif p}{\delta} \notag\\
      &= \int_{\Delta} \mu_w \left( \frac{p - p_0}{\delta} + p_0 \right) \frac{\dif p}{\delta}
    \end{align}
    this implies
    \begin{equation}
        \label{eq:35}
        \int_{\Delta} \mu_w \left( \frac{p - p_0}{\delta} + p_0 \right) \frac{\dif p}{\delta}
        > 1- \frac{\epsilon}{4} .
    \end{equation}

    To simplify the following equations, let us define $f(p) = \mu_w ( \frac{p - p_0}{\delta} + p_0)$.
    Therefore, for any $\epsilon$ we have
    \begin{align}
        \label{eq:36}
        &\left\| \frac{1}{\delta } \int_\Real f(p)
          \nu(p) \dif p - \nu(p_0) \right\|_1 \notag\\
        &\quad= \left\| \int_\Real
          f(p) \left[ \nu(p) - \nu(p_0) \right] \frac{\dif p}{\delta}
        \right\|_1 \notag\\
        &\quad\leq \int_\Real f(p)
        \left\| \nu(p) - \nu(p_0) \right\|_1 \frac{\dif p}{\delta} \notag\\
        &\quad\leq \int_\Delta f(p)
        \left\| \nu(p) - \nu(p_0) \right\|_1 \frac{\dif p}{\delta}
        + 2 \int_{\Real/\Delta}\!\! f(p) \frac{\dif p}{\delta} \notag\\
        &\quad\leq \int_\Delta f(p) \frac{\dif p}{\delta}
        \max_{p'\in \Delta} \left\|\nu(p') - \nu(p_0) \right\|_1
        + \frac{\epsilon}{2} \notag\\
        &\quad\leq \epsilon,
    \end{align}
    where the third inequality is due to $\left\| \nu(p) - \nu(p_0) \right\|_1 \leq \left\| \nu(p) \right\|_1 +\left\| \nu(p_0) \right\|_1 = 2$.
\end{proof}

Finally, combining \cref{eq:48} with \cref{thr:1} while setting $\nu(p) = \Tr_{w} \left[ {U_e(p)\rho_{e}(0)U_e^\dagger(p)} \right]$ leads to the statement in \cref{eq:32}.

\subsection{Mixture of Unitaries}
\label{sec:mixture-unitaries}

Given an initial state of the form $\rho(0) = \rho_{ub}(0)\otimes \rho_w(0)\otimes \rho_c(0)$, we calculate the reduced subsystem-bath state evolving under the Hamiltonian
\begin{align}
  \label{eq:38}
  H &= \hu + \hb + \hw + \hc + \hi\notag\\ 
    &= H_0 + \hi\notag\\ 
  \hi &= \iint_{\Real^2} \hi_{ub}(x,p) \otimes
        \ket{p}_\subw\!\bra{p}\otimes \ket{x}_\subc\!\bra{x}\dif p \dif x. \\ \end{align}
In the interaction picture $\hi$ takes the form
\begin{align}
  &\thi(t) = \expo{iH_0t}\hi\expo{-iH_0t} \\
  &\quad= \expo{iH_0t}\iint_{\Real^2}\!\!\! \hi_{ub}(x,p) \otimes
    \ket{p}_\subw\!\bra{p}\otimes \ket{x}_\subc\!\bra{x}\dif p \dif x \expo{-iH_0t}\notag \\
  &\quad= \iint_{\Real^2}\!\!\! \thi_{ub}(x,p,t) \otimes
    \ket{p-t}_\subw\!\bra{p-t}\otimes \ket{x-t}_\subc\!\bra{x-t}\dif p \dif x \notag\\ 
  &\quad= \iint_{\Real^2}\!\!\! \thi_{ub}(x+t,p+t,t) \otimes
    \ket{p}_\subw\!\bra{p}\otimes \ket{x}_\subc\!\bra{x}\dif p \dif x,  \notag
\end{align}
where
\begin{equation}
    \label{eq:5}
      \thi_{ub}(x,p,t)= \expo{i(\hu+\hb)t}\hi_{ub}(x,p)\expo{-i(\hu+\hb)t}.
\end{equation}
With this, the time evolution operator between times $0$ and $t$ can be written as
\begin{align}
  \mathcal{U}(t) &= \expo{-iH_0t}\timeorder{\expo{-i \int_{0}^{t} \thi(t') \dif t'}} \\
                 &= \expo{-i \hw t}\expo{-i \hc t}\expo{-i \hu t}\expo{-i \hb t} \notag\\
                 &\quad\quad\times\iint_{\Real^2}\!\!\!
                   \timeorder{\expo{-i \int_{0}^{t} \thi_{ub}(x+t',p+t',t') \dif t'}} \notag\\
                 &\quad\quad\quad\quad\times \ket{p}_\subw\!\bra{p}
                   \otimes \ket{x}_\subc\!\bra{x} \dif p \dif x,
\end{align}
When calculating the reduced time-evolved subsystem-bath state, the $\hw$ and $\hc$ exponentials vanish by ciclicity of the trace, and one is left with a mixture of unitaries.
That is,
\begin{align}
  \Tr_{wc}[\rho(t)]
  &= \Tr_{wc}[\mathcal{U}(t) \rho(0) \mathcal{U}^\dagger(t)] \\
  &= \iint_{\Real^2} \Tr_{wc}\left[(\rho_w(0)\otimes \rho_c(0)) (\ket{p}_\subw\!\bra{p}\otimes \ket{x}_\subc\!\bra{x}) \right] \notag\\
  &\quad\quad\quad\quad \times V(x,p) \rho_{ub}(0) V^\dagger(x,p) \dif p \dif x, \notag\\
  &= \iint_{\Real^2}  \mu(x, p) V(x,p) \rho_{ub}(0) V^\dagger(x,p) \dif p \dif x, \notag
\end{align}
where $\mu(x,p) = \braket{x,p|\rho_{cw}(0)|x,p}$ and
\begin{align}
  \label{eq:27}
  V(x,p) &=\expo{-i \hu t}\expo{-i \hb t} \timeorder{\expo{-i \int_{0}^{t} \thi_{ub}(x+t',p+t',t') \dif t'}}.
\end{align}
\end{document}